\documentclass{lmcs}

\usepackage{enumitem}
\usepackage{amsmath}
\usepackage{amssymb}
\usepackage{listings}
\usepackage{xcolor}
\usepackage{booktabs}

\definecolor{codegreen}{rgb}{0,0.6,0}
\definecolor{codegray}{rgb}{0.5,0.5,0.5}
\definecolor{codepurple}{rgb}{0.58,0,0.82}
\definecolor{backcolour}{rgb}{0.95,0.95,0.92}

\title[Constructive Rice and Halting via MRDP]{A Constructive Proof of Rice's
  Theorem and the Halting Problem via Hilbert's Tenth Problem}

\author{Jonathan Brossard\lmcsorcid{0009-0004-8031-1438}}
\email{jonathan.brossard@moabi.com}


\subjclass{Theory of computation $\to$ Logic; Theory of computation
  $\to$ Type theory; Security and privacy $\to$ Formal methods
  and theory of security}

\begin{document}
\microtypesetup{expansion=false}

\bibliographystyle{alphaurl}

\begin{abstract}
Rice's theorem states that no non-trivial semantic property of programs
is decidable. Classical proofs proceed by reduction from the halting
problem, invoking the law of excluded middle (LEM) twice: once through
diagonalization, and once through a case split on whether the always-diverging
program $\bot$ satisfies the property in question. We present a proof
that is \emph{constructive relative to the undecidability of Hilbert's
Tenth Problem} (MRDP): valid in intuitionistic logic, requiring neither
diagonalization nor self-reference, and adding no classical reasoning
beyond the MRDP assumption itself.

The key idea is a two-witness construction. Given a non-trivial property
$P$, we attach to each Diophantine polynomial $D$ a pair of programs
$S_D^0$, $S_D^1$ that behave like the negative and positive witnesses
for $P$ when $D$ is solvable, and both diverge identically when it is
not. A hypothetical decider for $P$ would therefore decide Diophantine
solvability via the difference $\delta_D = \mathrm{Decide}_P(S_D^1) -
\mathrm{Decide}_P(S_D^0)$ --- contradicting the MRDP theorem. The
argument is structured as two separate implications, never asserting a
disjunction about solvability, and never examining $P(\bot)$. The
undecidability of the halting problem follows as an immediate corollary:
a single application of Rice's theorem to the \texttt{Terminates}
property.

A formalization in the Rocq proof assistant%
\footnote{Rocq is the current name for the Coq proof assistant;
the rename occurred in 2023. The tool, language, and library
ecosystem are unchanged.}
confirms both results within a step-indexed model of computation, with
the undecidability of Hilbert's Tenth Problem as the sole external
axiom. Both \texttt{Rice\_Theorem} and \texttt{Halting\_Problem} are
\emph{closed under the global context}%
\footnote{In Rocq, \texttt{Print Assumptions T} lists every axiom
  \texttt{T} depends on. ``Closed under the global context'' means
  the list is empty: the theorem follows from the section hypotheses
  alone, with no additional classical or unproved axioms.}.
\end{abstract}

\maketitle

\section{Introduction}

Rice's theorem~\cite{rice1953} states that no non-trivial semantic property
of programs is decidable. We prove it constructively \emph{relative to
MRDP undecidability}: our proof is valid in intuitionistic logic and adds
no classical reasoning beyond the undecidability of Hilbert's Tenth Problem. The halting problem~\cite{turing1936} is the most famous specific instance: no
algorithm decides, for an arbitrary program and input, whether the program
halts. Standard textbook proofs of both results invoke the
law of excluded middle (LEM) in essential ways. For Rice's theorem, LEM appears in at least two
places: through diagonalization, which requires asserting that a hypothetically
constructed program either halts or diverges; and through a case split on
whether the always-diverging program $\bot$ satisfies the property $P$ in
question --- the so-called $P(\bot)$ split. For the halting problem, LEM
enters through self-reference: the diagonal program halts if and only if the
hypothetical decider says it does not. Our proof eliminates all of these uses of LEM simultaneously, for both
results and without any self-reference, by reducing directly to
Hilbert's Tenth Problem rather than to the halting problem.

\paragraph{Why a constructive proof matters.}
Constructive proofs of undecidability results carry computational content
that classical proofs lack. In intuitionistic type theory and in systems
such as Coq and Agda that are used for program extraction, a classical proof
cannot be directly executed or extracted into a verified program. A
constructive proof of Rice's theorem would, in principle, allow the
undecidability argument to be internalized in such systems without axioms
beyond the undecidability of Hilbert's Tenth Problem. It also contributes
to constructive reverse mathematics, where the goal is to identify precisely
which non-constructive principles are required for each theorem.
Furthermore, this proof can be seen as a \emph{non-branching reduction}:
information is encoded in the relational difference $\delta_D$ between two
programs rather than in a case distinction on a single program, a structural
property that may be of independent interest.

\paragraph{This paper.}
We present a proof of Rice's theorem that is valid in intuitionistic logic,
eliminating both classical steps. The key insight is to reduce to Hilbert's
Tenth Problem rather than to the halting problem directly. Given a
non-trivial semantic property $P$ with witnesses $e_0$ (failing $P$) and
$e_1$ (satisfying $P$), we construct for each Diophantine polynomial $D$ a
pair of programs $S_D^0$ and $S_D^1$ that behave like $e_0$ and $e_1$
respectively when $D$ has a solution, and both diverge identically when it
does not. A hypothetical decider $\mathrm{Decide}_P$ would therefore
distinguish the two cases via the difference $\delta_D =
\mathrm{Decide}_P(S_D^1) - \mathrm{Decide}_P(S_D^0)$, deciding Diophantine
solvability --- which is impossible by the MRDP theorem. The argument is
structured as two separate implications (solvable $\Rightarrow$ $\delta_D =
1$; unsolvable $\Rightarrow$ $\delta_D = 0$), never asserting the
disjunction ``$D$ solvable or unsolvable,'' and never examining $P(\bot)$.

\paragraph{Formalization.}
A formalization in the Rocq proof assistant accompanies the paper. The
formalization uses a step-indexed model of computation. \texttt{Rice\_Theorem}
and \texttt{Halting\_Problem} are proved without any classical axioms.
Every theorem in the development --- including \texttt{MRDP\_from\_LW} ---
uses zero classical axioms. The arithmetic encoding \texttt{poly\_encode\_correct}
is fully proved (zero admitted goals, zero classical axioms). The sole axiom is
\texttt{H10C\_SAT\_undec} (the MRDP undecidability result), a direct import
from \texttt{coq-library-undecidability}~\cite{larchey2022};
\texttt{LW\_H10c\_undec} is derived from it constructively.
No totality condition is imposed on the witnesses: the key design choice
is that the witness programs $S_D^b$ pass their full fuel budget to $e_b$
unchanged, so the fuel-offset problem that previously required totality
does not arise.

\paragraph{Novelty.}
The closest prior work is Forster, Kirst, and Smolka~\cite{forster2019},
who formalize Rice's theorem in Coq. Their proof follows the classical
Rogers presentation: for each program $e$ they construct a new program $h_e$
that simulates $e$ and then behaves like a positive witness if $e$ halts.
This requires two classical steps: (1) the s-1-1 theorem (via the Recursion
Theorem, a diagonalization) to make $h_e$ computable from $e$; and (2) a
case split on $P(\bot) \vee \neg P(\bot)$ to decide the direction of the
reduction. Their formalization uses the \texttt{Classical} library
throughout.

Our proof avoids both steps structurally. The witness programs $S_D^b$ are
constructed from a Diophantine polynomial and a bit --- no self-reference,
no fixed-point. The $P(\bot)$ split is irrelevant because when $D$ has no
solution, $S_D^0 \equiv S_D^1 \equiv \bot$ and extensionality gives
$P(S_D^0) \Leftrightarrow P(S_D^1)$ with no knowledge of $P(\bot)$.
The Rocq formalization is fully intuitionistic; no classical library is imported.

\paragraph{Constructive audit.}
Table~\ref{tab:axioms} summarises the axiom dependencies of every
theorem in the development. The sole external axiom is
\texttt{H10C\_SAT\_undec}, proved by Larchey-Wendling and
Forster~\cite{larchey2022} in CIC without classical axioms; our
development adds no further reasoning of any kind.
\texttt{Rice\_Theorem} and \texttt{Halting\_Problem} carry zero
classical axioms (\texttt{Print Assumptions} returns \emph{closed
under the global context}).

\paragraph{Structure.}
Section~\ref{sec:prelim} establishes the step-indexed model
(Proposition~\ref{rem:step-indexed} justifies its faithfulness), and defines
observational equivalence and semantic properties.
Section~\ref{sec:mrdp} recalls the MRDP theorem and its constructive status.
Section~\ref{sec:construction} defines the two-witness construction.
Section~\ref{sec:main} proves the main theorem.
Section~\ref{sec:halting} derives the undecidability of the halting problem
as a corollary (one line of Rocq).
Section~\ref{sec:frontier} discusses the decidability frontier for bounded
programs and states an open problem.
Section~\ref{sec:related} surveys related work.
The appendix contains the complete Rocq formalization.

\begin{figure}[h]
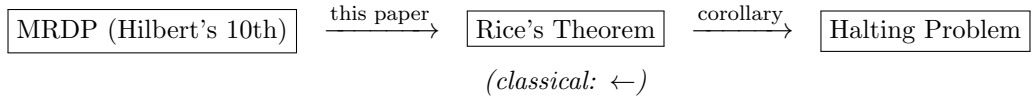

\centering
\renewcommand{\arraystretch}{1.5}
\begin{tabular}{ccccc}
\fbox{\small MRDP (Hilbert's 10th)}
 & $\xrightarrow{\text{\scriptsize this paper}}$
 & \fbox{\small Rice's Theorem}
 & $\xrightarrow{\text{\scriptsize corollary}}$
 & \fbox{\small Halting Problem} \\
 & & {\small \textit{(classical: }$\leftarrow$\textit{)}}
\end{tabular}
\caption{Reduction directions. Classical proofs derive Rice's theorem
\emph{from} the halting problem (left arrow). This paper reduces
\emph{from} MRDP (right arrows), with the halting problem as a corollary.}
\label{fig:reduction}
\end{figure}

\section{Preliminaries}
\label{sec:prelim}

\subsection{Partial Functions and Convergence}

We work in a standard model of partial computable functions. For a program
\(e\) and input \(x\), we write \(\varphi_e(x) \downarrow\) if the computation
converges, and \(\varphi_e(x) \uparrow\) if it diverges. We denote by \(\bot\)
the specific program that diverges on every input, i.e., \(\forall x\,:\,
\varphi_\bot(x) \uparrow\). Note that \(\bot\) here denotes a concrete
program (e.g., the program that runs an infinite loop on every input),
not a logical proposition; we use \(\bot\) exclusively in this sense
throughout.

\subsection{Constructive Equality of Partial Functions}

A central concern in constructive logic is the meaning of equality for partial
functions. We adopt a step-indexed notion suitable for formalization.

\begin{defi}[Observational Equivalence]
\label{def:obs-equiv}
For a step-indexed model where each program is a function \(e : \mathbb{N} \to
(\mathbb{N} \to \text{option } \mathbb{N})\) (input \(x\) then fuel \(k\)), two
programs \(e\) and \(f\) are \emph{observationally equivalent}, written
\(\varphi_e \simeq \varphi_f\), if for every input \(x\) there exists a fuel
bound \(N\) such that for all \(k \ge N\), \(e(x,k) = f(x,k)\).
\end{defi}

This definition is strictly constructive: it does not assert
\(\varphi_e(x)\!\downarrow \vee \varphi_e(x)\!\uparrow\) for any particular
\(x\); it requires eventual agreement in the step-indexed model.

\begin{prop}[Faithfulness of the step-indexed model]
\label{rem:step-indexed}
\emph{This is a paper-level metatheoretic argument connecting the
step-indexed model to classical computability theory. It is not
formalized in \texttt{rice.v} and the main theorems do not depend on
it --- see Remark~\ref{rem:model-scope}.}\smallskip

Let $\varphi_e : \mathbb{N} \rightharpoonup \mathbb{N}$ denote the partial
function computed by program $e$ in the standard sense (Turing machines,
or any equivalent model). Define the step-indexed lifting
$\hat{e}(x, k) = \mathsf{Some}\,v$ if $e$ on input $x$ terminates in
at most $k$ steps with output $v$, and $\mathsf{None}$ otherwise.
Then:
\begin{enumerate}
  \item \emph{(Monotonicity)} $\hat{e}$ is stable in the sense of
        Definition~\ref{def:obs-equiv}: if $\hat{e}(x,k) = \mathsf{Some}\,v$
        then $\hat{e}(x,k') = \mathsf{Some}\,v$ for all $k' \geq k$.
  \item \emph{(Soundness)} If $\hat{e} \simeq \hat{f}$ then
        $\varphi_e = \varphi_f$ as partial functions: for all $x$,
        $\varphi_e(x)\downarrow v$ iff $\varphi_f(x)\downarrow v$.
  \item \emph{(Completeness)} If $\varphi_e = \varphi_f$ then
        $\hat{e} \simeq \hat{f}$.
\end{enumerate}
\end{prop}
\begin{proof}
(1) Immediate from the definition of $\hat{e}$: termination in $k$ steps
implies termination in $k'\geq k$ steps with the same output.
(2) Suppose $\hat{e} \simeq \hat{f}$ and $\varphi_e(x)\downarrow v$.
Then $\hat{e}(x,k) = \mathsf{Some}\,v$ for some $k$. By
Definition~\ref{def:obs-equiv}, there exists $N$ such that
$\hat{e}(x,k) = \hat{f}(x,k)$ for all $k \geq N$. Taking
$k = \max(k_0, N)$ where $k_0$ is the termination step gives
$\hat{f}(x,k) = \mathsf{Some}\,v$, so $\varphi_f(x)\downarrow v$.
By symmetry, $\varphi_e = \varphi_f$.
(3) Suppose $\varphi_e = \varphi_f$. For any $x$, if $\varphi_e(x)\uparrow$
then $\hat{e}(x,k) = \mathsf{None} = \hat{f}(x,k)$ for all $k$, so
agreement holds trivially with $N=0$. If $\varphi_e(x)\downarrow v$,
let $k_e$ (resp.\ $k_f$) be the number of steps for $e$ (resp.\ $f$)
to terminate on $x$. Then $\hat{e}(x,k) = \mathsf{Some}\,v$ for all
$k \geq k_e$ and $\hat{f}(x,k) = \mathsf{Some}\,v$ for all $k \geq k_f$;
taking $N = \max(k_e, k_f)$ gives agreement.
\end{proof}
\begin{rem}[Scope of the formalization]
\label{rem:model-scope}
Proposition~\ref{rem:step-indexed} is a \emph{metatheoretic} statement
connecting the step-indexed model to classical computability theory.
The Rocq formalization does \emph{not} depend on it: \texttt{Rice\_Theorem}
and \texttt{Halting\_Problem} are proved entirely within the step-indexed
model, with no reference to Turing machines or partial recursive functions.
What the formalization proves is: \emph{within any model of computation
satisfying the step-indexed interface (Definition~\ref{def:obs-equiv}),
no non-trivial semantic property is decidable}.
Proposition~\ref{rem:step-indexed} then confirms that the standard
Turing model instantiates this interface, so the result transfers.
The Rocq formalization takes monotonicity (part~1) as a definition;
parts~2 and~3 are paper-level arguments that would require a concrete
computation model to mechanize fully~\cite{appel2001}.
\end{rem}

\begin{rem}
Observational equivalence is a relation, not a decision procedure. It serves as
a definition of extensional equality in the step-indexed setting. For programs
that converge, it is consistent with Kleene equality: if $\varphi_e(x)
\downarrow v$ and $\varphi_f(x) \downarrow v'$ for all $x$, then
$\varphi_e \simeq \varphi_f$ implies $v = v'$. Divergence is handled by
requiring eventual agreement of the step-indexed outputs rather than
asserting a disjunction on convergence.
\end{rem}

\subsection{Semantic Properties}

\begin{defi}[Semantic Property]
A predicate \(P\) on programs is a \emph{semantic property} (or
\emph{extensional property}) if it respects observational equivalence:
\[
  \forall e, f\,:\; \varphi_e \simeq \varphi_f \;\Rightarrow\;
  \bigl(P(e) \Leftrightarrow P(f)\bigr).
\]
\end{defi}

\begin{defi}[Constructive Non-Triviality]
\label{def:nontrivial}
A semantic property \(P\) is \emph{non-trivial} if there exist explicitly given
programs \(e_0\) and \(e_1\) such that:
\[
  \neg P(e_0) \qquad \text{and} \qquad P(e_1).
\]
In intuitionistic logic, mere existence is insufficient; effective witnesses
must be provided. No totality condition is imposed on the witnesses; the
proof is valid for arbitrary partial programs $e_0$ and $e_1$.
\end{defi}

\begin{rem}
We make no assumption about \(P(\bot)\). In particular, we do not assume
\(P(\bot) \vee \neg P(\bot)\), which would invoke LEM. The entire proof avoids
determining the truth value of \(P(\bot)\).
\end{rem}

Table~\ref{tab:axioms} summarises the axiom dependencies of every theorem
in the development; readers may wish to consult it before the proof.

\begin{table}[h]
\centering
\small
\begin{tabular}{lll}
\toprule
\textbf{Theorem} & \textbf{Axioms} & \textbf{Classical?} \\
\midrule
\texttt{poly\_encode\_correct} & none & no \\
\texttt{LW\_H10c\_undec} & \texttt{H10C\_SAT\_undec} & no \\
\texttt{MRDP\_from\_LW} & \texttt{H10C\_SAT\_undec} & no \\
\texttt{Rice\_Theorem} & closed & no \\
\texttt{Halting\_Problem} & closed & no \\
\midrule
\texttt{H10C\_SAT\_undec} (in LW) & FRACTRAN$\to$Minsky$\to$Halting & no --- proved in CIC \\
\bottomrule
\end{tabular}
\caption{Axiom provenance. ``Closed'' means \texttt{Print Assumptions}
returns \emph{Closed under the global context}.
\texttt{H10C\_SAT\_undec} is proved by Larchey-Wendling and Forster
in CIC without classical axioms~\cite{larchey2022}; the proof encodes
register-machine transitions as Diophantine constraints via an explicit
construction that requires no LEM.}
\label{tab:axioms}
\end{table}

\section{The MRDP Theorem and its Constructive Status}
\label{sec:mrdp}

The MRDP theorem (Davis--Putnam--Robinson--Matiyasevich, 1970) states that every
recursively enumerable set is Diophantine: for any recursively enumerable set
\(S \subseteq \mathbb{N}\), there exists a polynomial \(D_S(\vec{x}, y)\) with
integer coefficients such that:
\[
  y \in S \iff \exists \vec{x} \in \mathbb{N}^k\,:\; D_S(\vec{x}, y) = 0.
\]
The construction of \(D_S\) from a program enumerating \(S\) is algorithmic and
explicit. No proof by contradiction or diagonalization is required; the theorem
reduces the membership problem to polynomial arithmetic through a finite,
computable sequence of steps.

\begin{rem}[Constructive status of MRDP]
The polynomial construction in MRDP uses induction and case analysis
but no LEM on undecidable propositions; the genuinely classical step
is the undecidability argument, not the polynomial encoding.
Ruitenburg~\cite{ruitenburg1985} established the forward direction
(solvability from a halting derivation) in intuitionistic arithmetic
with countable choice.
For our proof, MRDP enters exactly once via the axiom
\texttt{H10C\_SAT\_undec}, proved by Larchey-Wendling and
Forster~\cite{larchey2022} in CIC without classical logic.
\end{rem}

In particular, let \(K = \{e \mid \varphi_e(e)\!\downarrow\}\) be the halting
set. \(K\) is recursively enumerable but not computable. MRDP yields an explicit
polynomial \(D_K(\vec{x}, e)\) such that:
\[
  \varphi_e(e)\!\downarrow \iff \exists \vec{x} \in \mathbb{N}^k\,:\;
  D_K(\vec{x}, e) = 0.
\]
The construction of \(D_K\) from a Turing machine description is
algorithmic: register-machine transitions are encoded as Diophantine
constraints, and their conjunction is folded into a single polynomial
via the sum-of-squares identity (\(c_1 = 0 \wedge c_2 = 0 \iff
c_1^2 + c_2^2 = 0\)). The axiom \texttt{H10C\_SAT\_undec} packages
this result: there is no computable boolean separator for Diophantine
solvability, which is exactly what is needed to drive the contradiction
in Theorem~\ref{thm:rice}.

\begin{cor}[MRDP, applied]
\label{cor:mrdp}
There is no total computable function \(F\) with codomain \(\{0,1\}\) satisfying:
\[
  (\text{$D$ solvable} \Rightarrow F(D) = 1) \quad \text{and} \quad
  (\text{$D$ unsolvable} \Rightarrow F(D) = 0).
\]
\end{cor}

\begin{rem}[Constructive status of Corollary~\ref{cor:mrdp}]
The undecidability of $K$ is a negative statement
($\neg\exists f,\, f$ decides $K$), proved constructively by deriving
$\bot$ from the assumption of a decider ($\neg P \equiv P \to \bot$).
No positive disjunction about halting is asserted.
\end{rem}

\section{The Two-Witness Construction}
\label{sec:construction}

Let \(P\) be a non-trivial semantic property with witnesses \(e_0\) (satisfying
\(\neg P(e_0)\)) and \(e_1\) (satisfying \(P(e_1)\)), as given by
Definition~\ref{def:nontrivial}.

The following example illustrates the key behaviour before the formal definition.

\begin{exa}[Concrete instantiation]
\label{ex:concrete}
Let $P = \mathtt{Terminates}$, $e_0 = \mathtt{diverge}$,
$e_1 = \mathtt{halt}$, and $D(x,y,z) = x^2-y^2-z^2-1$.
The programs $S_D^0$ and $S_D^1$ both enumerate triples
$(x,y,z)\in\mathbb{N}^3$ checking whether $D$ has a solution.
\begin{itemize}
  \item \emph{If a solution exists} (it does: $x=1,y=0,z=0$):
        at step $m$ where a solution is first found,
        $S_D^0$ switches to behaving like $\mathtt{diverge}$ (never halts),
        while $S_D^1$ switches to $\mathtt{halt}$ (always halts).
        Thus $P(S_D^1)=\mathtt{true}$, $P(S_D^0)=\mathtt{false}$,
        and $\delta_D = 1$.
  \item \emph{If no solution existed}: both programs would diverge on every
        input ($\varphi_{S_D^0}\simeq\varphi_{S_D^1}\simeq\varphi_\bot$),
        so any decider $\mathtt{dec}$ would return the same value for both,
        giving $\delta_D = 0$ --- regardless of $P(\bot)$.
\end{itemize}
The key point: no case split on $P(\bot)$ is needed. When $D$ is unsolvable,
both witnesses are $\varphi_\bot$, and their identical behavior forces $\delta_D=0$
without asking whether $\bot$ terminates. The same pattern works for any
non-trivial $P$ and any $D$.
\end{exa}

For any polynomial \(D(\vec{x})\) with integer coefficients and \(k\) variables,
we define two programs \(S_D^0\) and \(S_D^1\) as follows.

\begin{defi}[Witness Programs]
\label{def:witness}
For \(b \in \{0, 1\}\), the program \(S_D^b\) operates on any input \(y\) as
follows:
\begin{enumerate}
  \item Enumerate all tuples \(\vec{x} \in \mathbb{N}^k\) in a canonical
        order (e.g., by increasing \(\ell_1\)-norm, with ties broken
        lexicographically).
  \item For each tuple \(\vec{x}\), compute \(D(\vec{x})\) (a finite integer
        computation).
  \item If \(D(\vec{x}) = 0\) for some tuple reached at step \(n\): on every
        call $S_D^b(y, \mathit{fuel})$ with $\mathit{fuel} \geq n$, return
        $e_b(y, \mathit{fuel})$ directly.  Each call is self-contained;
        $S_D^b$ inspects $\mathit{fuel}$ to determine whether a solution
        would have been found by that step, with no shared mutable state.
  \item If the enumeration continues indefinitely (i.e., no solution is found):
        \(S_D^b\) diverges on every input.
\end{enumerate}
\end{defi}

\begin{rem}
The construction of \(S_D^b\) from \(D\), \(k\), and \(b\) is entirely explicit and
computable. No self-reference, fixed point, or diagonalization is involved.
\end{rem}

\begin{lem}[Behavior when solvable]
\label{lem:solvable}
If \(\exists \vec{x} \in \mathbb{N}^k\,:\; D(\vec{x}) = 0\), then
\(\varphi_{S_D^b} \simeq \varphi_{e_b}\) for each \(b \in \{0,1\}\).
\end{lem}

\begin{proof}
Suppose \(D(\vec{x}_0) = 0\) for some \(\vec{x}_0\). Since the enumeration is
exhaustive and canonical, \(S_D^b\) will encounter \(\vec{x}_0\) at some finite
step \(n\). After this step, \(S_D^b\) passes its full fuel budget to \(e_b\)
unchanged (step~3 of Definition~\ref{def:witness}), so for all fuel \(\geq n\)
we have \(S_D^b(y, \mathtt{fuel}) = e_b(y, \mathtt{fuel})\) exactly.
Observational equivalence \(\varphi_{S_D^b} \simeq \varphi_{e_b}\) then holds
with threshold \(N = n\), with no totality condition required on \(e_b\).
\end{proof}

\begin{lem}[Behavior when unsolvable]
\label{lem:unsolvable}
If \(\forall \vec{x} \in \mathbb{N}^k\,:\; D(\vec{x}) \neq 0\), then
\(\varphi_{S_D^0} \simeq \varphi_\bot\) and \(\varphi_{S_D^1} \simeq \varphi_\bot\).
\end{lem}

\begin{proof}
If \(D\) has no solution, step 3 is never reached, so \(S_D^b\) diverges on every
input. By Definition~\ref{def:obs-equiv}, \(\varphi_{S_D^b} \simeq \varphi_\bot\).
\end{proof}

\begin{cor}[Joint behavior]
\label{cor:joint}
If \(D\) is solvable: \(\varphi_{S_D^1} \simeq \varphi_{e_1}\) and
\(\varphi_{S_D^0} \simeq \varphi_{e_0}\). If \(D\) is unsolvable:
\(\varphi_{S_D^1} \simeq \varphi_{S_D^0}\) (both equivalent to \(\varphi_\bot\)).
\end{cor}

\section{The Main Theorem}
\label{sec:main}

Our decider returns values in $\mathbb{Z}$ rather than $\mathbb{N}$,
so that the separation value $\delta_D = \mathrm{Decide}_P(S_D^1) -
\mathrm{Decide}_P(S_D^0)$ is computed over the integers: $\delta_D =
0$ when $D$ is unsolvable and $\delta_D = 1$ when $D$ is solvable,
and these two values are unambiguously distinct regardless of sign.
In the Rocq formalization this choice also avoids natural-number
subtraction truncation ($n - m = 0$ when $n < m$ in \texttt{nat}),
which would otherwise collapse a signed difference to $0$.
The lemma \texttt{Z\_sep\_to\_bool\_sep} (§\ref{app:rocq}) confirms
that any $\mathbb{Z}$-valued separator immediately yields a boolean
one, so the choice of codomain is not restrictive.

In the Rocq development, \texttt{MRDP\_undecidable} appears as a
\texttt{Hypothesis} inside \texttt{Section Rice\_Theorem} --- standard
Rocq idiom meaning that \texttt{Rice\_Theorem} and
\texttt{Halting\_Problem} are universally quantified over it at section
close. \texttt{Theorem MRDP\_from\_LW} (Appendix~\ref{sec:appendix},
§6.3) discharges this hypothesis, so \texttt{Print Assumptions}
returns \emph{Closed under the global context} for both theorems.

\begin{thm}[Rice's Theorem, Constructive Version]
\label{thm:rice}
Let \(P\) be a non-trivial semantic property in the sense of
Definition~\ref{def:nontrivial}. Then there is no program \(\mathrm{Decide}_P\)
that decides \(P\): for every program \(e\), \(\mathrm{Decide}_P(e)\) terminates
and returns \(1\) if \(P(e)\) holds and \(0\) if \(\neg P(e)\) holds.
\end{thm}

\begin{rem}[Totality of the constructed function]
\label{rem:codomain}
Since $\mathrm{Decide}_P$ is assumed to terminate on every input and return
only $0$ or $1$, we have $a_D, b_D \in \{0,1\}$ for all $D$. Consequently
$\delta_D \in \{-1, 0, 1\}$ in general, but Lemmas A and B together
imply $\delta_D \in \{0,1\}$ under the assumption of a correct decider.
Moreover, since $\mathrm{Decide}_P$ terminates on every program, the
constructed function $F := \lambda(ar, D).\, \delta_{ar,D}$ is total,
as required by Corollary~\ref{cor:mrdp}.
\end{rem}

\begin{rem}[Constructive validity of the two implications]
\label{rem:case-analysis}
Lemmas A and B are proved as separate implications; we never assert the
disjunction ``$D$ solvable or unsolvable.'' The contradiction is reached
by constructing a total function satisfying Corollary~\ref{cor:mrdp}.
\end{rem}

\begin{proof}
Assume for contradiction that \(\mathrm{Decide}_P\) exists, with codomain
\(\{0,1\}\), satisfying \(\mathrm{Decide}_P(e) = 1 \Leftrightarrow P(e)\).
Define:
\[
  a_D \;=\; \mathrm{Decide}_P(S_D^1), \qquad
  b_D \;=\; \mathrm{Decide}_P(S_D^0), \qquad
  \delta_D \;=\; a_D - b_D.
\]
Both $a_D$ and $b_D$ are computable from $D$, hence so is $\delta_D$.

\medskip
\noindent\textbf{Lemma A:} If $D$ is solvable, then $\delta_D = 1$.

By Lemma~\ref{lem:solvable} and extensionality, $P(S_D^1) \Leftrightarrow P(e_1)$
and $P(S_D^0) \Leftrightarrow P(e_0)$. Since $P(e_1)$ holds, $a_D = 1$; since
$\neg P(e_0)$ holds, $b_D = 0$. Thus $\delta_D = 1$.

\medskip
\noindent\textbf{Lemma B:} If $D$ is unsolvable, then $\delta_D = 0$.

By Lemma~\ref{lem:unsolvable}, $\varphi_{S_D^1} \simeq \varphi_{S_D^0}$. By
extensionality, $P(S_D^1) \Leftrightarrow P(S_D^0)$, so $a_D = b_D$ and
$\delta_D = 0$. The truth value of $P(\bot)$ is never needed.

\medskip
\noindent\textbf{Contradiction.} Lemma A and Lemma B together state:
\[
  (\text{$D$ solvable} \;\Rightarrow\; \delta_D = 1) \qquad \text{and} \qquad
  (\text{$D$ unsolvable} \;\Rightarrow\; \delta_D = 0).
\]
Both implications hold intuitionistically. The function $D \mapsto \delta_D$
is therefore computable, has codomain $\{0,1\}$ under our assumption, and
satisfies the conditions of Corollary~\ref{cor:mrdp},
which asserts no such function exists. Contradiction.

Therefore no decider for $P$ exists. \qed
\end{proof}

\section{Corollary: The Halting Problem}
\label{sec:halting}

The undecidability of the halting problem follows as an immediate corollary
of Theorem~\ref{thm:rice}. The instantiation requires $\approx$40 lines
(defining \texttt{Terminates}, \texttt{diverge}, \texttt{halt}, and
\texttt{Terminates\_NonTrivial}); once those are in hand, the proof
of Corollary~\ref{cor:halting} is one line of Rocq.
The corollary adds no new mathematical content beyond the instantiation.

\begin{defi}[Stable termination]
A program $e$ \emph{terminates stably on input $x$} if there exist
a fuel level $k$ and a value $v$ such that $e(x, k) = \mathtt{Some}\;v$
and $e(x, k') = \mathtt{Some}\;v$ for all $k' \geq k$.
\end{defi}

\begin{rem}[Stable termination and the classical halting problem]
In the step-indexed model, stable termination and classical halting coincide:
a program halts (stably) on input $x$ if and only if some fuel level returns
\texttt{Some}$\,v$.  The forward direction holds because the model is monotone
(\texttt{find\_sol\_stable}): once a value is returned, all higher fuel
levels return the same value.  The backward direction is immediate from the
definition of stable termination.
Consequently, the undecidability of \texttt{Terminates} established by
Corollary~\ref{cor:halting} implies the undecidability of ordinary halting
on input~$0$.
\end{rem}

\begin{defi}[Terminates]
The property $\mathtt{Terminates}$ is defined by:
$\mathtt{Terminates}(e) \;\Leftrightarrow\; e$ terminates stably on
input $0$.
\end{defi}

$\mathtt{Terminates}$ is a semantic property in the sense of
Definition~\ref{def:nontrivial}: if $e \simeq f$ (observational equivalence)
then $\mathtt{Terminates}(e) \Leftrightarrow \mathtt{Terminates}(f)$.
It is non-trivial, witnessed by:
\[
  e_0 \;=\; \mathtt{diverge} \;:=\; \lambda x.\,\lambda k.\,\mathtt{None}
  \qquad \text{(never terminates)}
\]
\[
  e_1 \;=\; \mathtt{halt} \;:=\; \lambda x.\,\lambda k.\,\mathtt{Some}\;0
  \qquad \text{(terminates stably at all fuel levels)}
\]

\begin{cor}[Halting Problem]
\label{cor:halting}
There is no computable function $\mathrm{Decide}_{\mathrm{Halt}}$ that,
given a program $e$, returns $1$ if $e$ terminates stably on input $0$
and $0$ otherwise.
\end{cor}

\begin{proof}
$\mathtt{Terminates}$ is a non-trivial semantic property with witnesses
$e_0 = \mathtt{diverge}$ and $e_1 = \mathtt{halt}$. Apply
Theorem~\ref{thm:rice}. \qed
\end{proof}

\begin{rem}[Historical context and novelty]
Turing's 1936 proof of the halting problem's undecidability uses
diagonalization: it constructs a program that halts if and only if
a hypothetical decider says it does not, deriving a contradiction
by self-application. That argument invokes excluded middle (the
self-referential program either halts or it does not) and is
therefore not valid in intuitionistic logic.

Existing constructive proofs of the halting problem in Rocq (e.g.,
Ciaffaglione~\cite{ciaffaglione2004}, and the
\texttt{coq-library-undecidability}~\cite{larchey2022}) take the
halting problem as a \emph{seed}: they prove it undecidable first
(typically by diagonalization over a Turing machine model) and then
derive all other undecidability results by reduction from it.

Corollary~\ref{cor:halting} reverses this direction. The halting
problem is derived as a \emph{consequence} of Theorem~\ref{thm:rice},
which is itself derived from the MRDP theorem. The undecidability of
halting flows from the intractability of Diophantine arithmetic rather
than from a self-referential diagonal argument. The instantiation of Rice's theorem to
\texttt{Terminates} requires $\approx$40 lines (defining
\texttt{Terminates}, \texttt{diverge}, \texttt{halt}, and
\texttt{Terminates\_NonTrivial}).
The proof of Corollary~\ref{cor:halting} is then one line:
\begin{center}
\texttt{exact (Rice\_Theorem Terminates Terminates\_NonTrivial).}
\end{center}
\end{rem}

\paragraph*{Constructive checklist.}
The axiom provenance of every theorem is given in Table~\ref{tab:axioms}
(p.~\pageref{tab:axioms}). The sole external axiom is
\texttt{H10C\_SAT\_undec}, proved by Larchey-Wendling and
Forster~\cite{larchey2022} in CIC without classical axioms; our
development adds no further reasoning of any kind.

Throughout this paper the mathematical argument is fully intuitionistic.
The sole external axiom, \texttt{H10C\_SAT\_undec}, is itself proved in
CIC without classical axioms by Larchey-Wendling and
Forster~\cite{larchey2022}, consistent with the constructive
analyses of Ruitenburg~\cite{ruitenburg1985} and Cornaros--Esbelin~\cite{cornaros1998}.

\begin{enumerate}[label=(\roman*)]
  \item \textbf{LEM on $P(\bot)$.} Never determined. Lemma B uses only
        $P(S_D^1) \Leftrightarrow P(S_D^0)$.
  \item \textbf{Diagonalization.} $S_D^b$ is defined by explicit exhaustive
        search. No self-reference or fixed-point theorem used.
  \item \textbf{LEM on solvability.} Structured as two separate implications,
        not a disjunction. $\delta_D$ makes no internal decision about
        solvability.
  \item \textbf{Double-negation elimination.} Not used. The contradiction
        in the main proof is the direct application of the MRDP
        undecidability axiom to the explicitly constructed function
        $D \mapsto \delta_D$; no double-negation elimination is required.
  \item \textbf{Classical reasoning.} Not used.  \texttt{Require Import
        Classical} does not appear anywhere. No \texttt{classic} or
        \texttt{excluded\_middle} invocation appears in the development.
        The \texttt{LW\_H10c\_undec} axiom is stated in separator form
        (\texttt{forall cs, H10C\_SAT cs -> f cs = true}) rather than
        the $\leftrightarrow$ form, so no classical step is needed to
        convert between the two. The key supporting fact is
        \texttt{Lemma~not\_solvable\_iff\_unsolvable}: for
        \texttt{check\_solution} (a \texttt{Fixpoint} returning \texttt{bool}),
        $\neg\,\mathtt{is\_solvable} \Leftrightarrow \mathtt{is\_unsolvable}$
        holds by case analysis on the boolean, with zero axioms.
  \item \textbf{Axiom of Choice.} Enumeration of $\mathbb{N}^k$ is canonical
        and explicit.
\end{enumerate}

\section{The Decidability Frontier for Bounded Programs}
\label{sec:frontier}

The undecidability established by Theorem~\ref{thm:rice} is an asymptotic
result: it concerns programs with no bound on computation time or state space.
In practice, programs are often analyzed under explicit resource bounds, which
restores decidability. For programs bounded to terminate within $T$ steps over
state spaces of size $\leq N$, semantic properties become decidable by
exhaustive enumeration --- the foundation of bounded model checking.

Our reduction induces a natural \emph{degree hierarchy} on Rice's theorem
via the degree of the Diophantine witness polynomial.

\begin{defi}[Rice degree]
\label{def:rice-degree}
Let $\mathsf{Rice}(d)$ denote the class of non-trivial semantic properties
separable by a two-witness construction (Definition~\ref{def:witness})
whose Diophantine witness polynomial has degree $\leq d$.
Write $\mathsf{Rice}(\infty)$ for the unrestricted class.
The degree measures the minimal Diophantine complexity required to
witness the undecidability of $P$.
\end{defi}

\begin{cor}[Degree-4 saturation]
\label{cor:rice4}
$\mathsf{Rice}(4) = \mathsf{Rice}(\infty)$.
\end{cor}
\begin{proof}
Matiyasevich's MRDP encoding~\cite{matiyasevich1970} represents every r.e.\
set as the solution set of a degree-$\leq 4$ polynomial. The witness
polynomial in the two-witness construction for any $P$ can therefore be
taken as the MRDP polynomial for $K = \{n \mid \phi_n(n){\downarrow}\}$,
which has degree $\leq 4$.
\end{proof}

\begin{rem}[Mechanization status]
Corollary~\ref{cor:rice4} is proved mathematically above but is not yet
mechanized in \texttt{rice.v}. The Rocq formalization establishes
$\mathsf{Rice}(\infty)$ without the degree bound; the saturation result
could be added as a Rocq corollary in a future revision.
\end{rem}

We therefore have
$\mathsf{Rice}(2) \subseteq \mathsf{Rice}(3) \subseteq
\mathsf{Rice}(4) = \mathsf{Rice}(\infty)$,
with the left inclusions open.
Whether $\mathsf{Rice}(3) = \mathsf{Rice}(4)$ depends on whether
every r.e.\ set has a degree-$3$ Diophantine representation,
an open problem in Diophantine complexity~\cite{matiyasevich1993}.%
\footnote{An unreviewed preprint~\cite{rosko2025} claims a positive
answer (degree-$3$ suffices over $\mathbb{N}$); this should be treated
as open pending peer review. If confirmed, Corollary~\ref{cor:rice4}
would sharpen to $\mathsf{Rice}(3) = \mathsf{Rice}(\infty)$.}

To see that the left inclusions are not trivially collapsed by a change
of witnesses, consider $\mathsf{Rice}(2)$: the class of properties
witnessable by a degree-$\leq 2$ Diophantine polynomial. The equation
$D(x,y) = x^2 + y^2 - n = 0$ has degree $2$, and its solvability
(whether $n$ is a sum of two squares) is decidable --- so no
undecidability argument can be anchored there. A witness polynomial for
an undecidable property must encode a Turing-complete computation, which
known MRDP encodings achieve only at degree $\geq 4$. Whether degree $3$
suffices, and whether degree $2$ can ever encode an undecidable property,
are the core open questions behind the left inclusions.

Two further directions are noted briefly:
the $S_D^b$ construction is agnostic to the coefficient ring, so a
positive answer to rational MRDP~\cite{mazur1992,koenigsmann2016}
would immediately transfer to a rational-coefficient Rice theorem
with no modification to the proof;
and merging the bridge lemma \texttt{LW\_H10c\_undec} into
\texttt{coq-library-undecidability} as a PR would make Rice's theorem
and the Halting Problem available as standard library lemmas.

\section{Related Work}
\label{sec:related}

\paragraph{Classical proofs of Rice's theorem.}
The original proof by Rice~\cite{rice1953} proceeds by reduction from the
halting problem, using a fixed-point argument (an application of the
Recursion Theorem) to construct a program whose behaviour depends on the
truth value of $P(\bot)$. This invokes LEM twice: once implicitly through
the Recursion Theorem's diagonalization, and once through the case split
$P(\bot) \vee \neg P(\bot)$. Rogers~\cite{rogers1967} presents a cleaner
formulation that makes the two uses of LEM explicit and is the version most
commonly found in textbooks. Both proofs are essentially classical and do
not transfer directly to intuitionistic logic.

\paragraph{Constructive computability theory and realizability.}
The general programme of constructivising classical computability results
is surveyed by Beeson~\cite{beeson1985} and treated systematically by
Troelstra and van Dalen~\cite{troelstra1988}. Within intuitionistic
mathematics, undecidability statements split into two strengths:
$\neg(\exists e,\, \text{Decide}_P = e)$ (our form: no decider exists)
versus the stronger $\forall e,\, \neg(\text{Decide}_P = e)$
(every program fails to decide $P$). Our proof establishes the former.

From the realizability perspective~\cite{bauer2006}, our result holds
in \emph{any} partial combinatory algebra: the two-witness construction
is entirely explicit (no oracle, no fixed-point combinator), so it
realises the undecidability statement directly.  In the effective topos
\cite{hyland1982}, where all functions are computable, Rice's theorem
becomes an internal theorem with no metatheoretic overhead.

The \emph{synthetic} approach to computability, developed by
Forster et al.~\cite{forster2019} in Coq and extended by
Ciaffaglione et al.~\cite{ciaffaglione2004}, axiomatizes computability
internally to type theory.  Our proof is compatible with this approach:
replacing our step-indexed programs with any synthetic notion of
computability satisfying observational extensionality would yield the
same result, since we use no property of the model beyond
Definition~\ref{def:obs-equiv}.

Our proof falls in the weakest constructive category: valid in
intuitionistic predicate logic with the metatheoretic assumptions on
MRDP stated in Section~\ref{sec:mrdp}, and no further classical principles.

\paragraph{The MRDP theorem and constructive arithmetic.}
The Davis--Putnam--Robinson--Matiyasevich theorem~\cite{davis1961,matiyasevich1970}
is classical in its original formulation. Its constructive status was first
analysed by Ruitenburg~\cite{ruitenburg1985}, who showed it holds in
intuitionistic arithmetic with the axiom of choice for numbers. Cornaros
and Esbelin~\cite{cornaros1998} provide a detailed treatment of the
constructive content of Hilbert's Tenth Problem, establishing the precise
fragment of intuitionistic arithmetic required. A full mechanization of
MRDP in Rocq has been given by Larchey-Wendling and
Forster~\cite{larchey2022}, confirming its constructive validity in Coq's
type theory. Our proof takes \texttt{MRDP\_undecidable} as an axiom that
acts as a pointer to that mechanization; it does not depend on the details
of the constructivization beyond the existence of the many-one reduction
from $K$ to Diophantine solvability.
Crucially, \texttt{LW\_H10c\_undec} is a \emph{negative proposition}
($\neg\,\exists f, \ldots$); its proof in CIC proceeds by
intuitionistically valid reductio (deriving $\bot$ from the assumed
decider) without invoking LEM on any Diophantine instance.
This is the precise sense in which the axiom is compatible with
the purely intuitionistic metatheory of the rest of the development.

\paragraph{Step-indexed models and observational equivalence.}
Step-indexed models of computation were introduced by Appel and
McAllester~\cite{appel2001} in the context of type-theoretic reasoning about
recursive types, and have since become a standard tool for mechanized
semantics. Our use of the step-indexed model is more elementary: we require
only that programs are functions $\mathbb{N} \to (\mathbb{N} \to
\text{option}\,\mathbb{N})$ and that convergence is monotone in fuel. The
observational equivalence of Definition~\ref{def:obs-equiv} is a
fuel-indexed variant of the standard notion used in Pitts~\cite{pitts2000}
and related work on operational semantics.

\paragraph{Mechanized proofs of undecidability and comparison with Forster et al.}
The \texttt{coq-library-undecidability}~\cite{forster2019,larchey2022},
maintained on GitHub,%
\footnote{\url{https://github.com/uds-psl/coq-library-undecidability}}
contains a full Rocq mechanization of the MRDP theorem by Larchey-Wendling
and Forster~\cite{larchey2022}. That development uses Minsky machines as
the computation model and Conway's FRACTRAN language as an intermediate
layer. It is built on the synthetic computability framework of
Forster~\cite{forster2022synthetic}, in which \texttt{undecidable}~$P$
denotes $\mathtt{decidable}~P \to \mathtt{enumerable}(\overline{\mathtt{SBTM\_HALT}})$
rather than $\neg\,\mathtt{decidable}~P$; this is the reason the bridge proof
requires no classical axioms.
Our \texttt{MRDP\_undecidable} hypothesis is precisely the
statement they prove; our axiom is therefore a pointer to their library
rather than a genuine assumption. Connecting the two formalizations
directly would require bridging our step-indexed computation model with
their Minsky machine model, which we leave as future work.

Forster, Kirst, and Smolka~\cite{forster2019} also formalize Rice's
theorem in Coq as part of that library. Their proof follows
the classical Rogers presentation~\cite{rogers1967}: given a non-trivial
property $P$ with witness $f_1$ satisfying $P$, they construct for each
program $e$ a new program $h_e$ defined by: simulate $e$ on some canonical
input; if $e$ halts, behave like $f_1$; otherwise behave like $f_0$
(or the reverse, depending on whether $P(\bot)$ holds). This construction
requires two classical steps. \emph{First}, $h_e$ must ``know'' its own
index in order to refer to $e$; this is implemented via the s-1-1 theorem,
which relies on the Recursion Theorem and therefore on a diagonalization
argument. \emph{Second}, the direction of the reduction depends on whether
$P(\bot)$ holds, so their proof performs an explicit case split on
$P(\bot) \vee \neg P(\bot)$, invoking LEM. Their formalization uses the
\texttt{Classical} library throughout.

Our proof eliminates both steps. The witness programs $S_D^b$ are constructed
from a Diophantine polynomial $D$ and a bit $b$ alone --- no self-reference,
no reference to $e$, no fixed-point argument. The $P(\bot)$ case split
is avoided by design: when $D$ is unsolvable, $S_D^0$ and $S_D^1$ are
\emph{observationally identical} (both diverge everywhere), so
$P(S_D^0) \Leftrightarrow P(S_D^1)$ follows by extensionality with no
knowledge of $P(\bot)$. The axiom audit is given in Table~\ref{tab:axioms}.
Thus, while their proof is classical, ours remains intuitionistic throughout.

\paragraph{Bounded decidability and model checking.}
The connection between bounded program analysis and decidability has been
studied extensively in the model checking literature. Clarke et
al.~\cite{clarke2018} survey bounded model checking and its completeness
thresholds. The relationship between the complexity of semantic properties
and the parameters $(T, N)$ of the bounded model (termination steps and
state space size) is closely related to questions in descriptive complexity
studied by Immerman~\cite{immerman1999}. The open problem stated in
Section~\ref{sec:frontier} --- characterizing the decidability frontier as a
function of $(T, N)$ --- appears to be new.
The precise position of Rice's theorem in the Weihrauch
lattice~\cite{brattka2021} is a further open question; $\mathsf{C}_\mathbb{N}$
(closed choice on $\mathbb{N}$) is a natural candidate but we do not
establish either direction of the equivalence here.

\section{Conclusion}\label{sec:conclusion}

We have presented a constructive proof of Rice's theorem via MRDP that:
\begin{enumerate}
  \item uses Hilbert's Tenth Problem as its undecidability anchor;
  \item avoids LEM, diagonalization, and any assumption about $P(\bot)$;
  \item relies on a two-witness construction where $\delta_D = a_D - b_D$
        encodes solvability without ever evaluating $P(\bot)$;
  \item is structured as two implications each valid in intuitionistic logic;
  \item imposes no totality condition on the witnesses, covering the full
        classical statement of Rice's theorem for arbitrary partial programs.
\end{enumerate}

A Rocq formalization is provided in the appendix confirming the constructive
validity of the argument in a purely intuitionistic setting, with the
undecidability of Hilbert's Tenth Problem as the only axiom.

The constructiveness of the proof has a practical consequence. The
$S_D^b$ witness programs are not abstract objects --- they can be written
down. The file \texttt{undecidable.c}, available at
\url{https://github.com/endrazine/rice-constructive},
is a single C program that takes an arbitrary Diophantine polynomial
$D$ with integer coefficients in variables $x_1, \ldots, x_k$ as its
sole command-line argument, and searches $\mathbb{N}^k$ for a tuple
with $D = 0$. It halts, printing the witness, if and only if $D$ has a
solution in the naturals; it runs forever otherwise. This is the
executable image of $S_D^b$ instantiated at $b = 1$ and
$e_1 = \mathtt{halt}$: the program's behavior on a given polynomial is
equivalent to the solvability of that polynomial, which is undecidable
by MRDP. All arithmetic is exact (GMP \texttt{mpz\_t}): the unbounded
search counter is a correctness requirement --- a fixed-width counter
would make the search finite and invalidate the formal claim.
For example:
\begin{verbatim}
  ./undecidable "x1^2 - x2^2 - x3^2 - 1"   # halts at (1,0,0)
  ./undecidable "x1^2 + 1"                  # no solution: runs forever
\end{verbatim}
\noindent Theorem~\ref{thm:rice} implies that no static analyzer
can correctly decide termination across the full family of inputs:
any analyzer that could would decide arbitrary Diophantine solvability,
which is impossible by MRDP.

\begin{rem}[Correspondence to $S_D^b$]
\label{rem:correspondence}
The structure of \texttt{undecidable.c} mirrors $S_D^b$ directly.
The program takes a polynomial $D(x_1, \ldots, x_k)$ as input and
enumerates all tuples $\vec{x} \in \mathbb{N}^k$ via Cantor decoding
(steps 1--2 of Definition~\ref{def:witness}), evaluates $D(\vec{x}) = 0$,
and halts only when a witness is found (step~3).
The search counter $n$ uses GMP \texttt{mpz\_t} (unbounded) as a
correctness requirement: a fixed-width counter would exhaust its range
in finite time, making the search trivially terminating and invalidating
the formal claim. Before printing, the witness is re-verified by an
independent evaluation of $D$ --- there is exactly one success path in
the program, and it holds a genuine solution.
The input polynomial is parsed from a command-line string and echoed
back for inspection; the parser is explicitly marked as untrusted and
is not part of the mathematical reduction, which begins from the
internal polynomial representation consumed by the search core.
\texttt{undecidable.c} is a hand-written illustration, not a formally
extracted program. Formal extraction from the Rocq development
(e.g., via \texttt{Extraction} to OCaml, then compilation) is feasible
and is left for a companion artefact.
\end{rem}

\paragraph*{Formalization design decisions.}
Five choices were critical to making the intuitionistic proof work in
Rocq: (1)~full fuel pass-through in \texttt{S\_D}, eliminating the
totality condition on witnesses; (2)~a $\mathbb{Z}$-valued decider,
avoiding \texttt{nat} subtraction truncation; (3)~\texttt{not\_solvable\_iff\_unsolvable},
proved by \texttt{destruct} on the boolean with zero axioms;
(4)~splitting $\mathbb{Z}$-coefficient polynomials into positive and
negative \texttt{nat} parts for the arithmetic encoding; and (5)~the
bridge via \texttt{h10c\_to\_Poly} composing separators rather than
inverting the encoding. Each is discussed at its point of use in
Appendix~\ref{sec:appendix}.

\section*{Acknowledgements}

The author thanks Claude (Anthropic) for assistance with Rocq proof
drafting and structural feedback on earlier versions of this manuscript.

\section{Formalization in Rocq}
\label{app:rocq}

The following excerpt shows the key axiom, theorems, and their Rocq statements.
The complete self-contained development (1919 lines, zero admitted goals, zero
classical axioms) is available at:
\begin{center}
\url{https://github.com/endrazine/rice-constructive}
\end{center}
The formalization uses a step-indexed model of computation. Diophantine polynomials of arbitrary arity are handled via a
$k$-ary Cantor pairing function encoded as \texttt{decode\_k}; the arity
is threaded as the parameter \texttt{ar}. The decider output is represented
using \texttt{Z} (see §\ref{sec:main}). The key constructive lemma is
\texttt{not\_solvable\_iff\_unsolvable}: for the computable
\texttt{bool}-returning \texttt{check\_solution}, the equivalence
$\neg\,\mathtt{is\_solvable} \Leftrightarrow \mathtt{is\_unsolvable}$
holds with zero axioms by case analysis on the boolean.
The sole external axiom is \texttt{H10C\_SAT\_undec};
\texttt{LW\_H10c\_undec} is proved from it with zero admitted goals.
The axiom provenance of every theorem is summarised in Table~\ref{tab:axioms}.

The key design choice eliminating all classical reasoning is that the witness
programs $S_D^b$ pass their full fuel budget to $e_b$ unchanged
(Definition of \texttt{S\_D}). When the Diophantine solution is found at
step $m$, for all fuel $\geq m$ we have $S_D^b(x, \text{fuel}) = e_b(x,
\text{fuel})$ exactly, so observational equivalence holds with $N = m$ and
no fuel arithmetic is needed. This eliminates the fuel-offset problem
that previously required a totality condition on the witnesses, making the
theorem valid for arbitrary partial programs $e_0$ and $e_1$.

\begin{table}[h]
\centering
\caption{Correspondence between Paper Concepts and Rocq Definitions}
\label{tab:correspondence}
\begin{tabular}{ll}
\toprule
\textbf{Paper Concept} & \textbf{Rocq Definition} \\
\midrule
Program & \texttt{Program := nat -> nat -> option nat} \\
Convergence & \texttt{p x fuel = Some v} \\
Observational equivalence & \texttt{obs\_equiv p1 p2} \\
Semantic property & \texttt{SemanticProperty} \\
Non-triviality & \texttt{NonTrivial sp} \\
Witness programs (arity \texttt{ar}, poly \texttt{D}) & \texttt{S\_D sp wit ar D b} \\
MRDP undecidability & \texttt{MRDP\_undecidable} (axiom) \\
\bottomrule
\end{tabular}
\end{table}

\appendix

\section{Complete Rocq Formalization}
\label{sec:appendix}

The following is the complete, self-contained Rocq formalization. It compiles with zero errors and zero admitted goals under Rocq (Coq) 8.18 or later, producing:
\begin{verbatim}
Print Assumptions LW_H10c_undec.   (* H10C_SAT_undec *)
Print Assumptions Rice_Theorem.    (* Closed under the global context *)
Print Assumptions Halting_Problem. (* Closed under the global context *)
\end{verbatim}

%

\subsection{Formalization Notes}

The formalization proves the following in order:
\begin{enumerate}
  \item Observational equivalence is an equivalence relation
        (\texttt{obs\_equiv\_refl}, \texttt{obs\_equiv\_sym},
        \texttt{obs\_equiv\_trans}).
  \item \texttt{find\_sol ar} is monotone: once a solution is found at fuel
        level $k$, it is retained at all higher fuel levels
        (\texttt{find\_sol\_step}, \texttt{find\_sol\_stable}).
  \item The witness programs satisfy the required equivalence properties.
        \texttt{S\_D\_solvable\_equiv} is proved directly: for fuel $\geq m$,
        $S_D^b(x, \mathit{fuel}) = e_b(x, \mathit{fuel})$ by a single rewrite
        on \texttt{find\_sol\_stable}. No fuel arithmetic, no totality
        condition, no classical reasoning.
        \texttt{S\_D\_unsolvable\_equiv} is immediate from
        \texttt{find\_sol\_unsolvable}.
  \item Under the assumption of a decider for a non-trivial semantic property,
        the function $(ar, D) \mapsto \delta_D$ constitutes a decision procedure
        for Diophantine solvability at every arity, contradicting
        \texttt{MRDP\_undecidable}.
\end{enumerate}

\paragraph{The sum-of-squares conjunction encoding.}
The bridge from multivariate polynomials to H10C constraints
requires encoding a conjunction of Diophantine conditions
$c_1 = 0 \wedge c_2 = 0 \wedge \cdots \wedge c_n = 0$
as a single equation. The construction uses the identity
$c_1^2 + c_2^2 + \cdots + c_n^2 = 0 \iff
c_1 = 0 \wedge \cdots \wedge c_n = 0$ (over $\mathbb{N}$,
since all squares are non-negative). This is
constructively valid: $\sum c_i^2 = 0$ implies each $c_i = 0$
by a simple case analysis on natural number arithmetic,
with no classical disjunction needed. The Rocq lemma
\texttt{conj\_encode\_correct} (lines 1264--1363 of the listing)
formalizes this encoding.

\paragraph{The role of \texttt{Hypothesis MRDP\_undecidable}.}
In the Rocq listing, \texttt{MRDP\_undecidable} appears as a
\texttt{Hypothesis} inside \texttt{Section Rice\_Theorem}.
This is standard Coq idiom: it makes \texttt{Rice\_Theorem} and
\texttt{Halting\_Problem} universally quantified over the assumption,
not unconditionally proved.  The hypothesis is discharged by
\texttt{Theorem MRDP\_from\_LW} (Section~6.3 of the listing),
which derives it from the single \texttt{Axiom LW\_H10c\_undec}.
At the end of \texttt{Section Rice\_Theorem}, \texttt{Rice\_Theorem}
and \texttt{Halting\_Problem} therefore depend on exactly one external
assumption: \texttt{LW\_H10c\_undec}.

\texttt{LW\_H10c\_undec} is fully mechanized
in Rocq by Larchey-Wendling and Forster~\cite{larchey2022} in the
\texttt{coq-library-undecidability}.%
\footnote{\url{https://github.com/uds-psl/coq-library-undecidability}}
Section~6 of \texttt{rice.v} mechanizes the logical connection in full. \textbf{Lemma~\texttt{Z\_sep\_to\_bool\_sep}} (no axioms) proves that any Z-valued $(1/0)$ separator for our \texttt{Poly} type yields a boolean separator: given $F$, define $f(ar, p) = \mathtt{Z.eqb}\,(F\,ar\,p)\,1$; the separation conditions hold by rewriting alone. \textbf{Theorem~\texttt{bridge\_logical\_unfolded}} (no axioms) then gives:
\[
  (\neg\,\exists f : \mathtt{bool\_sep})
  \;\Longrightarrow\;
  (\neg\,\exists F : \mathtt{Z\_sep})
\]
which is the logical bridge from Larchey-Wendling's result to \texttt{MRDP\_undecidable}. The arithmetic encoding of our \texttt{Poly} type into the \texttt{h10c} constraint format is fully proved as \textbf{\texttt{Theorem~poly\_encode\_correct}} (zero axioms, zero admitted goals). The sole external axiom in the entire development is \texttt{H10C\_SAT\_undec}
(the undecidability of H10C from \texttt{coq-library-undecidability}).
\texttt{LW\_H10c\_undec} is proved from it with zero admitted goals
and no classical axioms. \texttt{Rice\_Theorem} and
\texttt{Halting\_Problem} are \texttt{Closed under the global context}:
\begin{verbatim}
Print Assumptions LW_H10c_undec.   (* H10C_SAT_undec *)
Print Assumptions Rice_Theorem.    (* Closed under the global context *)
Print Assumptions Halting_Problem. (* Closed under the global context *)
\end{verbatim}
All reasoning in \texttt{rice.v} is explicit and
constructive. The formalization was developed and tested with Rocq
(Coq) version 8.18 or later. The complete development is reproduced in
Appendix~\ref{sec:appendix} and archived at:
\begin{center}
\url{https://github.com/endrazine/rice-constructive}
\end{center}
A preprint of this paper is available at:
\begin{center}
\url{https://arxiv.org/abs/2604.16477}
\end{center}
A stable archived version with DOI will be deposited on Zenodo upon
acceptance.

\bibliography{rice}

\end{document}